\documentclass[11pt, onecolumn]{article}

\newtheorem{theorem}{Theorem}

\newtheorem{obs}{Observation}

\newtheorem{prop}{Proposition}

\def\QED{\ensuremath{{\square}}}
\def\markatright#1{\leavevmode\unskip\nobreak\quad\hspace*{\fill}{#1}}
\newenvironment{proof}
  {\begin{trivlist}\item[\hskip\labelsep{\bf Proof.}]}
  {\markatright{\QED}\end{trivlist}}

\usepackage{times,graphics,epsfig,amsmath,shapepar,fancyhdr,amssymb}
\usepackage{verbatim}

\title{Morphing of Triangular Meshes in Shape Space
\thanks{Research supported in part by HPCVL and NSERC.}}

\author{Stefanie Wuhrer\thanks{Carleton University, Ottawa, Canada, and National Research Council of Canada, Ottawa, Canada, \textit{swuhrer@scs.carleton.ca}.}
\and Prosenjit Bose\thanks{Carleton University, Ottawa, Canada, \textit{jit@scs.carleton.ca}.}
\and Chang Shu\thanks{National Research Council of Canada, Ottawa, Canada, \textit{chang.shu@nrc-cnrc.gc.ca}.}
\and Joseph O'Rourke\thanks{Smith College, Northampton, USA, \textit{orourke@cs.smith.edu}.}
\and Alan Brunton\thanks{University of Ottawa, Ottawa, Canada, and National Research Council of Canada, Ottawa, Canada, \textit{abrunton@site.uottawa.ca}.}}

\date{\today}

\begin{document}

\maketitle

\begin{abstract}
We present a novel approach to morph between two isometric poses of the same non-rigid object given as triangular meshes. We model the morphs as linear interpolations in a suitable shape space $\mathcal{S}$. For triangulated $3D$ polygons, we prove that interpolating linearly in this shape space corresponds to the most isometric morph in $\mathbb{R}^3$. We then extend this shape space to arbitrary triangulations in $3D$ using a heuristic approach and show the practical use of the approach using experiments. Furthermore, we discuss a modified shape space that is useful for isometric skeleton morphing. All of the newly presented approaches solve the morphing problem without the need to solve a minimization problem.
\end{abstract}
\section{Introduction}

Given two isometric poses of the same non-rigid object as triangular meshes $S^{(0)}$ and $S^{(1)}$ with known point-to-point correspondences, we aim to find a smooth isometric deformation between the poses. Interpolating smoothly between two given poses is called \textit{morphing}. We achieve this by finding shortest paths in an appropriate shape space similar to the approach by Kilian et al.~\cite{kilian_2007_gmss}. We propose a novel shape space.

A deformation of a shape represented by a triangular mesh is isometric if and only if all edge lengths are preserved during the deformation~\cite{kilian_2007_gmss}. This property holds because each face of the mesh is a triangle. A deformation of a shape is called \textit{most isometric} if the sum of the squared differences between the corresponding edge lengths of the two shapes is minimized. In this paper, we examine isometric deformations of general \textit{triangular meshes} in $3D$ and of \textit{triangulated $3D$ polygons}, which are triangular meshes with no interior vertices. We introduce a new shape space $\mathcal{S}$ for triangulated $3D$ polygons that has the property that interpolating linearly in shape space corresponds to the most isometric morph in $\mathbb{R}^3$. We then extend this shape space to arbitrary triangulations in $3D$ using a heuristic approach. Furthermore, we discuss a modification of the shape space that is useful for isometric skeleton morphing.

\section{Related Work}

Computing a smooth morph from one pose of a shape in two or three dimensions to another pose of the same shape has numerous applications. For example in computer graphics and computer animation this problem has received considerable attention~\cite{lazarus_verroust_98_morph_survey, alt_giubas_96_survey_shapes}. A recent survey on this topic was written by Alexa~\cite{alexa_02_survey}. We only review the work most relevant to this paper.

Before considering morphing three-dimensional mesh models, the two-dimensional version of the problem has received considerable attention. In the case where the input is sampled over a regular domain, this problem is called \textit{image morphing}. Image-morphing is widely studied and surveyed by Wolberg~\cite{wolberg_image_morph_survey}. In the case where the input is sampled over an irregular domain, the problem is to interpolate between two simple polygons in the plane. We only review work that makes use of intrinsic representations of the polygons. Sederberg et al.~\cite{sederberg_gao_wang_mu_93_2dmorph} propose to interpolate an intrinsic representation of two-dimensional polygons, namely the edge lengths and interior angles of the polygon. Surazhsky and Gotsman~\cite{surazhsky_gotsman_03_morph_2D} morph by computing mean value barycentric coordinates based on an intrinsic representation of triangulated polygons. This method is guaranteed to be intersection free. Iben et al.~\cite{iben_obrien_demaine_06_refolding_polygons} morph planar polygons while guaranteeing that no self-intersections occur using an approach based on energy minimization. This approach can be constrained to be as isometric as possible.

Sun et al.~\cite{sun_wang_chin_97_morph} morph between three-dimensional manifold meshes. They extend the approach by Sederberg et al.~\cite{sederberg_gao_wang_mu_93_2dmorph} to three dimensions by extending the intrinsic representation to polyhedra. However, the developed methods are computationally expensive~\cite{alexa_02_survey}. We propose a different intrinsic representation of triangular meshes that offers the advantage of producing the most isometric morph and of being efficient.

Sorkine and Alexa~\cite{sorkine_alexa_07_rigid_deformation} propose an algorithm to deform a surface based on a given triangular surface and updated positions of few feature points. The surface is modeled as a covering by overlapping cells. The deformation aims to deform each cell as rigidly as possible. The overlap is necessary to avoid stretching along cell boundaries. The deformation is based on minimizing a global non-linear energy function that is simple to implement. The energy is guaranteed to converge. However, since the energy function may have multiple minima, the algorithm is not guaranteed to find the global minimum. The approach tends to preserve the edge lengths of the triangular mesh. This property depends upon finding a global minimum of the energy function. One cannot guarantee to find this global minimum.

Recently, Zhou et al.~\cite{zhou_huang_snyder_liu_bao_guo_shum_07_volume_preserve} proposed a new method to deform triangular meshes based on the Laplacian of a graph representing the volume of the triangular mesh. The method is shown to prevent volumetric details to change unnaturally.

Recently, Kilian et al.~\cite{kilian_2007_gmss} used shape space representations to guide morphs and other more general deformations between shapes represented as triangular meshes. Each shape is represented by a point in a high-dimensional shape space and deformations are modeled as geodesics in shape space. The geodesic paths in shape space are found using an energy-minimization approach. Before Kilian et al.~\cite{kilian_2007_gmss} presented the use of a shape space for shape deformation and exploration of triangular meshes, shape space representations were developed to deform shapes in different representations. Cheng et al.~\cite{chen_edelsbrunner_fu_ShapeSpace} proposed an approach that deforms shapes given in skin representation, which is a union of spheres that are connected via blending patches of hyperboloids, with the help of a suitable shape space. Furthermore, algorithms for deforming curves with the help of shape space representations were proposed by Younes~\cite{younes_99_deformation} and Klassen et al.~\cite{klassen_srivastava_mio_04_planar_shape_spaces}. Eckstein et al.~\cite{eckstein_pons_tong_kuo_desbrun_07_surface_flow} propose a generalized gradient descent method similar to the approach by Kilian et al. that can be applied to deform triangular meshes. All of these approaches depend on solving a highly non-linear optimization problem with many unknown variables using numerical solvers. It is therefore not guaranteed that the globally optimal solution is found.

In this paper, we propose a novel shape space with the property that interpolating linearly in shape space approximates the most isometric morph in $\mathbb{R}^3$. For triangulated $3D$ polygons, we prove that the linear interpolation in shape space corresponds exactly to the most isometric morph in $\mathbb{R}^3$. For arbitrary triangulated manifolds in $3D$, we provide a heuristic approach to find the morph. This heuristics is an extension of the approach developed for $3D$ polygons. The proposed methods do not require solving a minimization problem.

\section{Theory of Shape Space for Triangulated $3D$ Polygons}
\label{polygon}

This section introduces a novel shape space for triangulated $3D$ polygons with the property that interpolating linearly in shape space corresponds to the most isometric morph in $\mathbb{R}^3$. The dimensionality of the shape space is linear in the number of vertices of the deformed polygon.

We start with two triangulated $3D$ polygons $P^{(0)}$ and $P^{(1)}$ corresponding to two near-isometric poses of the same non-rigid object. We assume that the point-to-point correspondence of the vertices $P^{(0)}$ and $P^{(1)}$ are known. Furthermore, we assume that both $P^{(0)}$ and $P^{(1)}$ share the same underlying mesh structure $M$. Hence, we know the mesh structure $M$ with two sets of ordered vertex coordinates $V^{(0)}$ and $V^{(1)}$ in $\mathbb{R}^3$, where $M$ is an outer-planar graph. We will show that we can represent $P^{(0)}$ and $P^{(1)}$ as points $p^{(0)}$ and $p^{(1)}$ in a shape space $\mathcal{S}$, such that each point $p^{(t)}$ that is a linear interpolation between $p^{(0)}$ and $p^{(1)}$ corresponds to a triangular mesh $P^{(t)}$ isometric to $P^{(0)}$ and $P^{(1)}$ in $\mathbb{R}^3$.

As we know the point-to-point correspondence of the vertices $P^{(0)}$ and $P^{(1)}$, we can find the best rigid alignment of the two shapes by solving an overdetermined linear system of equations and by modifying the solution to ensure a valid rotation matrix using an approach similar to the one used for camera calibration~\cite{trucco_verri_book}.

Let $M$ consist of $n$ vertices. As $M$ is a triangulation of a $3D$ polygon with $n$ vertices, $M$ has $2n-3$ edges and $n-2$ triangles. We assign an arbitrary but fixed order on the vertices, edges, and faces of $M$. The shape space $\mathcal{S}$ is defined as follows. The first $3$ coordinates of a point $p \in \mathcal{S}$ correspond to the coordinates of the first vertex $v$ in $M$. Coordinates $4$ and $5$ of $p$ correspond to the direction of the first edge of $M$ incident to $v$ in spherical coordinates. The next $2n-3$ coordinates of $p$ are the lengths of the edges in $M$ in order. The final $2(n-2)$ coordinates of $p$ describe the outer normal directions of the triangles in $M$ in spherical coordinates, in order. Hence, the shape space $\mathcal{S}$ has dimension $5+2n-3+2(n-2) = 4n-2=O(n)$. 

In the following, we prove that interpolating linearly between $P^{(0)}$ and $P^{(1)}$ in shape space yields the most isometric morph. To interpolate linearly in shape space, we interpolate the edge lengths by a simple linear interpolation. That is, $p_k^{(t)} = t p_k^{(0)} + (1-t) p_k^{(1)}$, where $p_k^{(x)}$ is the $k$th coordinate of $p^{(x)}$. The normal vectors are interpolated using geometric spherical linear interpolation (SLERP)~\cite{shoemake_85_slerp}. That is, $p_k^{(t)} = \frac{\sin(1-t)\Theta}{\sin \Theta} p_k^{(0)} + \frac{\sin t \Theta}{\sin \Theta} p_k^{(1)}$, where $\Theta$ is the angle between the two directions that are interpolated.

Note that the relative rigid alignment of $P^{(0)}$ and $P^{(1)}$ in $\mathbb{R}^3$ has an influence on the linear interpolation. That is, the interpolating shape space point varies as the relative rigid alignment of $P^{(0)}$ and $P^{(1)}$ in $\mathbb{R}^3$ changes. The change occurs because the angles between the normal vectors of $p^{(t)}$ change as a result of the rigid transformation. This is the reason we choose to find the best rigid alignment of $P^{(0)}$ and $P^{(1)}$ before transforming the polygons into $\mathcal{S}$.

To study interpolation in shape space, we make use of the \textit{dual graph $D(M)$} of $M$. The dual graph $D(M)$ has a node for each triangle of $M$. We denote the dual node corresponding to face $f$ of $M$ by $D(f)$. Two nodes of $D(M)$ are joined by an arc if the two corresponding triangles in $M$ share an edge. We denote the dual arc corresponding to an edge $e$ of $M$ by $D(e)$. Note that because $M$ meshes a $3D$ polygon, it is an outer-planar triangular graph and so the dual graph of $M$ is a binary tree. An example of a mesh $M$ with its dual graph $D(M)$ is shown in Figure~\ref{dual_graph}.

\begin{figure}[htb]
\centering
\includegraphics[width=5.0cm]{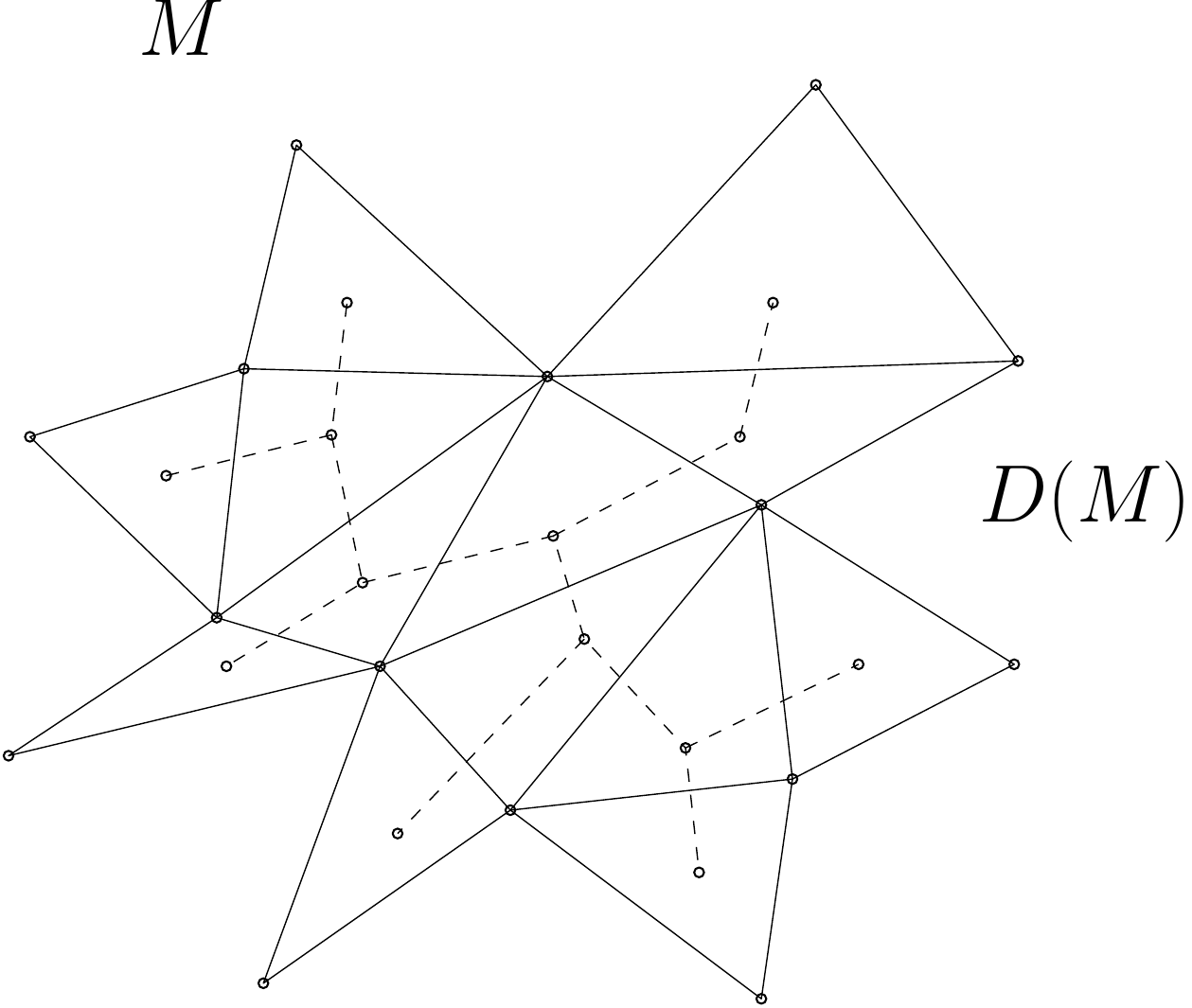}
\caption{\textit{A mesh $M$ with its dual graph $D(M)$.}}
\label{dual_graph}
\end{figure}

\begin{theorem}
Let $M$ be the underlying mesh structure of the triangulated $3D$ polygons $P^{(0)}$ and $P^{(1)}$. The linear interpolation $p^{(t)}$ between $p^{(0)}$ and $p^{(1)}$ in shape space $\mathcal{S}$ for $0\leq t\leq 1$ has the following properties:
\begin{enumerate}
\item The mesh $P^{(t)} \in \mathbb{R}^3$ that corresponds to $p^{(t)} \in \mathcal{S}$ is uniquely defined and has the underlying mesh structure $M$. We can compute this mesh using a traversal of the binary tree $D(M)$ in $O(n)$ time.
\item If $P^{(0)}$ and $P^{(1)}$ are isometric, then $P^{(t)}$ is isometric to $P^{(0)}$ and $P^{(1)}$. If $P^{(0)}$ and $P^{(1)}$ are not perfectly isometric, then each edge length of $P^{(t)}$ linearly interpolates between the corresponding edge lengths of $P^{(0)}$ and $P^{(1)}$.
\item The coordinates of the vertices of $P^{(t)}$ are a continuous function of $t$.
\end{enumerate}
\label{theorem_polygon}
\end{theorem}

\begin{proof}
\textbf{Part 1:} To prove uniqueness, we start by noting that the first vertex $v$ of $P^{(t)}$ is uniquely determined by the first three coordinates of $p^{(t)}$. The direction of the first edge $e$ of $M$ incident to $v$ is uniquely determined by coordinates $4$ and $5$ of $p^{(t)}$, because each point on the unit sphere determines a unique direction in $\mathbb{R}^3$. The length of each edge of $P^{(t)}$ is uniquely determined by the following $2n-3$ coordinates. Furthermore, the outer normal of each triangle is uniquely determined by the following $2(n-2)$ coordinates, because each point on the unit sphere determines a unique direction in $\mathbb{R}^3$. Hence, the edge $e$ is uniquely determined. For a triangle $f$ containing $e$, we now know the position of two vertices of $f$, the plane containing $f$, and the three lengths of the edges of $f$. Assuming that the normal vectors in shape space represent right-hand rule counterclockwise traversals of each triangle, this uniquely determines the position of the last vertex of $f$. We can now determine the coordinate of each vertex of $P^{(t)}$ uniquely by traversing $D(M)$. We start the traversal of $D(M)$ at $D(f)$. Recall that the coordinates of the vertices of triangle $f$ are known. Hence, when traversing an arc $D(e)$ incident to $D(f)$, we know the vertex coordinates of the shared edge between the two triangles corresponding to endpoints of $D(e)$. Denote the endpoint of $D(e)$ not corresponding to $f$ by $f'$. For $f'$, we now know the position of two vertices of $f'$, the plane containing $f'$, and the three lengths of the edges of $f'$. Hence, we can compute the position of the last vertex of $f'$. Because we now know the coordinates of all the vertices of $f'$, we can traverse all of the arcs in $D(M)$ incident to $D(f')$. In this fashion, we can set all of the vertex coordinates of $P^{(t)}$ by traversing $D(M)$. Because $D(M)$ is a tree, it is cycle-free. Hence, the coordinates of each vertex of $P^{(t)}$ are set exactly once. Because the complexity of $D(M)$ is $O(n)$, the algorithm terminates after $O(n)$ steps.

It remains to prove that $P^{(t)}$ is a valid triangular mesh, that is, that the three edge lengths of each triangle of $P^{(t)}$ satisfy the triangle inequality. We assume that both input meshes were valid triangular meshes. Hence, for any triangle $t$ with edge lengths $a^{(0)}, b^{(0)}, c^{(0)}$ in $P^{(0)}$ and $a^{(1)}, b^{(1)}, c^{(1)}$ in $P^{(1)}$, the following inequalities hold: $$a^{(0)}+b^{(0)}\geq c^{(0)}, b^{(0)}+c^{(0)}\geq a^{(0)}, c^{(0)}+a^{(0)}\geq b^{(0)}$$ and $$a^{(1)}+b^{(1)}\geq c^{(1)}, b^{(1)}+c^{(1)}\geq a^{(1)}, c^{(1)}+a^{(1)}\geq b^{(1)}.$$
In $P^{(t)}$, $a^{(t)}= (1-t)a^{(0)}+ta^{(1)}, b^{(t)}= (1-t)b^{(0)}+tb^{(1)}, c^{(t)}= (1-t)c^{(0)}+tc^{(1)}$ due to the linear interpolation of the end positions. Hence,
$a^{(t)} + b^{(t)} = (1-t)a^{(0)}+ta^{(1)} + (1-t)b^{(0)}+tb^{(1)} = (1-t)(a^{(0)} + b^{(0)}) + t(a^{(1)}+ b^{(1)}) \geq (1-t)c^{(0)} + t c^{(1)}.$ Similarly, we can show that $b^{(t)}+c^{(t)}\geq a^{(t)}$ and $c^{(t)}+a^{(t)}\geq b^{(t)}$. Hence, $P^{(t)}$ is a valid triangular mesh.

\textbf{Part 2:} The edge lengths of $P^{(t)}$ are linear interpolations between the edge lengths of  $P^{(0)}$ and $P^{(1)}$. Hence, the claim follows.

\textbf{Part 3:} When varying $t$ continuously, the point $p^{(t)} \in \mathcal{S}$ varies continuously. Hence, the coordinate of the lengths of all the edges vary continuously. Because a direction $ \left[\sin{v} \cos{u}, \sin{v} \sin{u}, \cos{v}\right]^T $ varies continuously if $u$ and $v$ vary continuously, the normal directions vary continuously. Because all the vertex positions of the mesh $P^{(t)}$ are uniquely determined by continuous functions of those quantities, all vertex positions of $P^{(t)}$ vary continuously.
\end{proof}

Note that we do not need to solve minimization problems to find the shortest path in shape space as in Kilian et al.~\cite{kilian_2007_gmss}. The only computation required to find an intermediate deformation pose is a graph traversal of $D(M)$.

Because $D(M)$ has complexity $O(n)$, we can traverse $D(M)$ in $O(n)$ time. Hence, we can compute intermediate deformation poses in $O(n)$ time each. We denote this the \textit{polygon algorithm} in the following.

\section{Theory of Shape Space for Skeleton Morphing}
A very similar shape space to the one presented in Section~\ref{polygon} can be used to isometrically morph between two topologically equivalent skeletons. Let a \textit{skeleton} in $\mathbb{R}^3$ be a set of joints connected by links arranged in a tree-structure. That is, we can consider a skeleton to be a tree in $\mathbb{R}^3$ consisting of $n$ vertices and $n-1$ edges.

The shape space presented in Section~\ref{polygon} can be simplified to a shape space for skeletons in $\mathbb{R}^3$ with the property that interpolating linearly in shape space corresponds to the most isometric morph in $\mathbb{R}^3$. The dimensionality of the shape space is linear in the number of links of the skeleton.

We start with two skeletons $S^{(0)}$ and $S^{(1)}$ corresponding to two near-isometric poses. We assume that the point-to-point correspondence of $S^{(0)}$ and $S^{(1)}$ are known. Hence, we know the tree structure $T$ with two sets of ordered vertex coordinates $V^{(0)}$ and $V^{(1)}$ in $\mathbb{R}^3$. As before, we first find the best rigid alignment of the two skeletons.

The skeleton shape space $\mathcal{S_S}$ is defined in a similar way as $\mathcal{S}$. We assign an arbitrary but fixed root to $T$ and traverse the edges of $T$ in a depth-first order. We assign an arbitrary order to the edges incident on each vertex of $T$. The first $3$ coordinates of $s$ correspond to the coordinates of the root in $T$. The next $n-1$ coordinates of $s$ are the lengths of the edges in $T$ in depth-first order. The final $2(n-1)$ coordinates of $s$ describe the unit directions of the edges in spherical coordinates in depth-first order. All edges are oriented such that they point away from the root. Note that the shape space $\mathcal{S}$ has dimension $O(n)$. 

Interpolating linearly between points in $\mathcal{S_S}$ is performed the same way as interpolating linearly between points in $\mathcal{S}$. Namely, edge lengths are interpolated linearly and unit directions are interpolated via SLERP. With a technique similar to the proof of Theorem~\ref{theorem_polygon}, we can prove the following theorem. In the proof, we do not need to consider a dual graph, but we can simply traverse the tree $T$ in depth-first order to propagate the information.

\begin{theorem}
Let $T$ be the underlying tree structure of the skeletons $S^{(0)}$ and $S^{(1)}$. The linear interpolation $s^{(t)}$ between $s^{(0)}$ and $s^{(1)}$ in shape space $\mathcal{S_S}$ for $0\leq t\leq 1$ has the following properties:
\begin{enumerate}
\item The skeleton $S^{(t)} \in \mathbb{R}^3$ that corresponds to $s^{(t)} \in \mathcal{S_S}$ is uniquely defined and has the underlying tree structure $T$. We can compute this tree using a depth-first traversal of the tree in $O(n)$ time.
\item If $S^{(0)}$ and $S^{(1)}$ are isometric, then $S^{(t)}$ is isometric to $S^{(0)}$ and $S^{(1)}$. If $S^{(0)}$ and $S^{(1)}$ are not perfectly isometric, then each edge length of $S^{(t)}$ linearly interpolates between the corresponding edge lengths of $S^{(0)}$ and $S^{(1)}$.
\item The coordinates of the vertices of $S^{(t)}$ are a continuous function of $t$.
\end{enumerate}
\label{theorem_skeleton}
\end{theorem}

This theorem allows us to morph isometrically between the skeletons of two shapes corresponding to two postures of the same articulated object in $O(n)$ time. Figure~\ref{skeleton} shows an example of such a morph.

\begin{figure}[htb]
\centering
\begin{tabular}{c c c}
\includegraphics[width = 4.5cm]{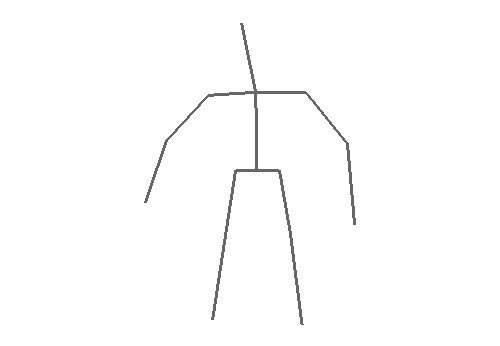} &
\includegraphics[width = 4.5cm]{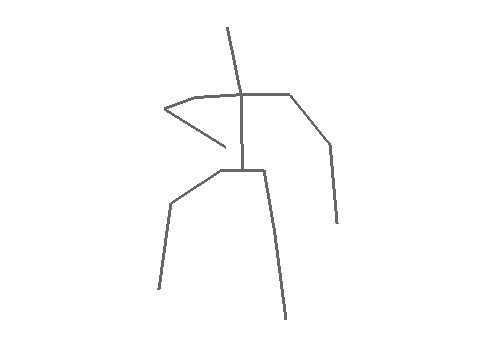} &
\includegraphics[width = 4.5cm]{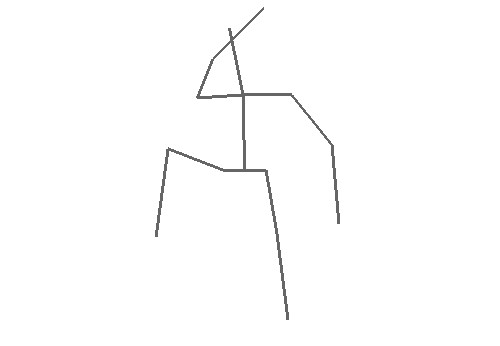} \\
(a) & (b) & (c)\\
\end{tabular}
\caption{\textit{Example of isometric morph between two skeleton poses. Input poses are shown in (a) and (c) and morph for $t=0.5$ is shown in (b).}}
\label{skeleton}
\end{figure}

In the remainder of this paper, we will focus our attention on morphing between triangular meshes.

\section{Generalization to Triangular Meshes}
\label{connected}

This section extends the shape space $\mathcal{S}$ from Section~\ref{polygon} to arbitrary connected triangular meshes. However, we can no longer guarantee the properties of Theorem~\ref{theorem_polygon}, because the dual graph of the triangular mesh $M$ is no longer a tree.

Given two triangular meshes $S^{(0)}$ and $S^{(1)}$ corresponding to two near-isometric poses of the same non-rigid object with known point-to-point correspondence, we know one mesh structure $M$ with two sets of ordered vertex coordinates $V^{(0)}$ and $V^{(1)}$ in $\mathbb{R}^3$. 

As before, we can use this information to find the best rigid alignment in $\mathbb{R}^3$. We do this before representing the shapes in a shape space. As outlined above, this alignment has a major influence on the result of the morph.

As before, we can represent $S^{(0)}$ and $S^{(1)}$ as points $s^{(0)}$ and $s^{(1)}$ in a shape space $\mathcal{S}$ using the same shape space points as in Section~\ref{polygon}. Let $s^{(t)}$ be the linear interpolation of $s^{(0)}$ and $s^{(1)}$ in $\mathcal{S}$, where the linear interpolation is computed as outlined in Section~\ref{polygon}. The existence of a mesh $S^{(t)} \in \mathbb{R}^3$ that has the underlying mesh structure $M$ and that corresponds to $s^{(t)}$ is no longer guaranteed. This can be seen using the example shown in Figure~\ref{example_joe}. Figure~\ref{example_joe}(a) and (b) show two isometric meshes $S^{(0)}$ and $S^{(1)}$. The dual graph $D(M)$ of the mesh structure $M$ is a simple cycle. Note that although the start and the end pose are isometric, we cannot find an intermediate pose that satisfies all of the interpolated normal vectors with SLERP and that is isometric to $S^{(0)}$ and $S^{(1)}$.

\begin{figure}[htb]
\centering
\begin{tabular}{c c}
\includegraphics[width = 1.5cm]{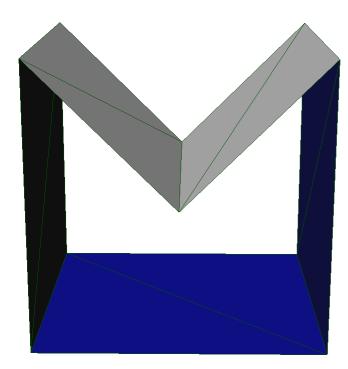} &
\includegraphics[width = 1.5cm]{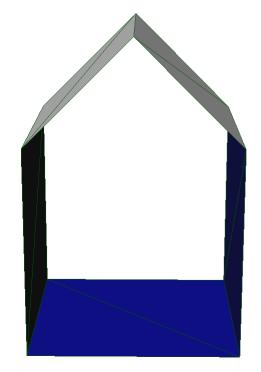} \\
(a) & (b) \\
\end{tabular}
\caption{\textit{Example of isometric triangular meshes where intermediate poses interpolating all normals and edge lengths do not exist.}}
\label{example_joe}
\end{figure}

Let $M$ consist of $n$ vertices. As $M$ is a planar graph, $M$ has $O(n)$ edges and $O(n)$ triangles. The shape space $\mathcal{S}$ is defined using the same shape space points as in Section~\ref{polygon}. The shape space $\mathcal{S}$ has dimension $O(n)$. As before, we interpolate linearly in shape space by interpolating the edge lengths by a simple linear interpolation.

\begin{obs}
Given a triangular mesh $S^{(t)}$ with underlying mesh structure $M$, point $s^{(t)}$ in $\mathcal{S}$ is uniquely determined. However, the inverse operation, that is computing a triangular mesh $S^{(t)}$ given a point $s^{(t)} \in \mathcal{S}$, is ill-defined.
\end{obs}

This is illustrated in Figure~\ref{example_joe}.

To compute a unique triangular mesh $S^{(t)}$ given a point $s^{(t)} \in \mathcal{S}$ that linearly interpolates between $s^{(0)}$ and $s^{(1)}$, such that $S^{(t)}$ approximates the information given in $s^{(t)}$ well, we use the dual graph $D(M)$ of $M$. Unlike in Section~\ref{polygon}, $D(M)$ is not necessarily a tree. Our algorithm therefore operates on a minimum spanning tree $T(M)$ of $D(M)$. The tree $T(M)$ is computed by assigning a weight to each arc $e$ of $D(M)$. The weight of $e$ is equal to the difference in dihedral angle of the supporting planes of the two triangles of $M$ corresponding to the two endpoints of $e$. That is, we compute the dihedral angle between the two supporting planes of the two triangles of $M$ corresponding to the two endpoints of $e$ for the start pose $S^{(0)}$ and for the end pose $S^{(1)}$, respectively. The weight of $e$ is then set as the difference between those two dihedral angles, which corresponds to the change in dihedral angle during the deformation. The weight can therefore be seen as a measure of rigidity. The smaller the weight, the smaller the change in dihedral angle between the two triangles during the deformation, and the more rigidly the two triangles move with respect to each other. As $T(M)$ is a minimum spanning tree, $T(M)$ contains the arcs corresponding to the most rigid components of $M$.

Similar to Section~\ref{polygon}, we compute $S^{(t)}$ by traversing $T(M)$. However, unlike in Section~\ref{polygon}, setting the vertex coordinates of a vertex $v$ of $S^{(t)}$ using two paths from the root of $T(M)$ to two triangles containing $v$ can yield two different resulting coordinates for $v$. An example of this situation is given in Figure~\ref{dual_mst}, where the coordinates of $v$ can be set by traversing the arcs $e_2$, and $e_3$ of $T(M)$ or by traversing the arcs $e_1, e_4$, and $e_5$ of $M$ starting at $root(T(M))$. We call the different coordinates computed for $v$ in $T(M)$ \textit{candidate coordinates} of $v$. Our algorithm computes the coordinates of each vertex $v \in S^{(t)}$ as the average of all the candidate coordinates of $v$.

\begin{figure}[htb]
\centering
\includegraphics[width=5.0cm]{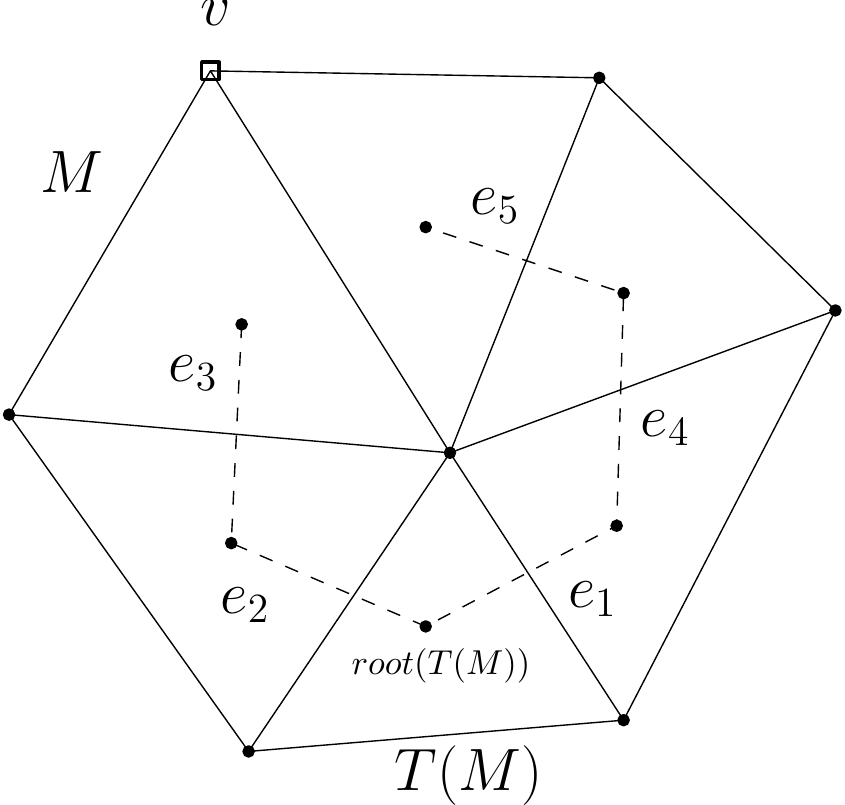}
\caption{\textit{A mesh $M$ with its dual minimum spanning tree $T(M)$.}}
\label{dual_mst}
\end{figure}

Let us analyze the maximum number of candidate coordinates that can occur for a vertex in $S^{(t)}$. Let $e$ denote an edge of $M$ such that $D(e)$ is in $T(M)$. Let $v$ denote the vertex of $S^{(t)}$ opposite $e$ in the triangle corresponding to an endpoint of $D(e)$, such that the coordinates of $v$ are computed when traversing $D(e)$. This situation is illustrated in Figure~\ref{num_candidates}. Let $d_1$ and $d_2$ denote the total number of candidate coordinates of the two endpoints of $e$. By traversing $D(e)$, we compute $d_1d_2$ candidate coordinates for $v$. We can therefore bound the number of candidate coordinates of $v$ computed using the path through $D(e)$ by $d_1d_2$. Note that the number of candidate coordinates for the two endpoints of the first edge is one. Furthermore, each vertex $v$ can be reached by at most $deg(v)$ paths in $T(M)$, where $deg(v)$ denotes the degree of vertex $v$ in $M$. As each path in $T(M)$ has length at most $m-1$, where $m=O(n)$ is the number of triangles of $M$, we can bound the total number of candidate coordinates in $S^{(t)}$ by $\sum_{v \in V} 2^{m-1}deg(v)=2n 2^{m-1}$, where $V$ is the vertex set of $M$.

\begin{figure}[htb]
\centering
\includegraphics[width=4.0cm]{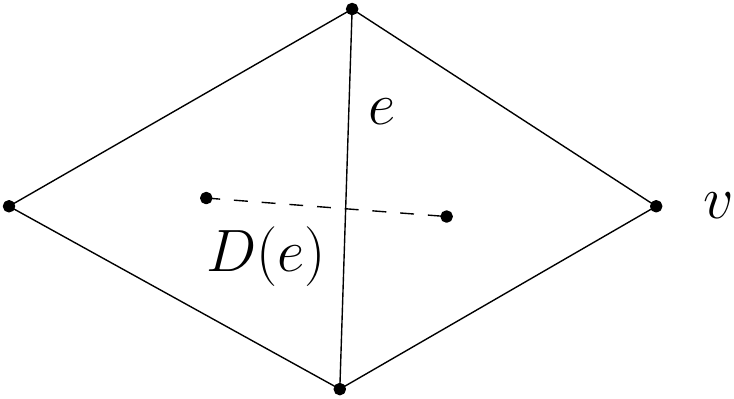}
\caption{\textit{Illustration of how to bound the number of candidate coordinates of $v$ computed using the path through $D(e)$.}}
\label{num_candidates}
\end{figure}

Our algorithm finds a triangular mesh $S^{(t)}$ corresponding to $s^{(t)}$ that is isometric to $S^{(0)}$ and $S^{(1)}$ if such a mesh exists, because all of the candidate coordinates are equal in this case and taking their average yields the desired result. If there is no isometric mesh corresponding to $s^{(t)}$, our algorithm finds a unique mesh that weighs all the evidence given by $T(M)$ equally. By choosing $T(M)$ as a minimum spanning tree based on weights representing rigidity, we allocate rigid parts of the model more emphasis than non-rigid parts. The reason for this is that in most near-isometric morphs, triangles close to non-rigid joints are deformed more than triangles in mainly rigid parts of the model. We conclude with the following Lemma.

\begin{prop}
Let $S^{(0)}$ and $S^{(1)}$ denote two isometric connected triangular meshes and let $s^{(0)}$ and $s^{(1)}$ denote the corresponding shape space points, respectively. We can compute a unique triangular mesh $S^{(t)}$ approximating the information given in the linear interpolation $s^{(t)}, 0\leq t\leq 1$ of $s^{(0)}$ and $s^{(1)}$, in exponential time. We find a triangular mesh $S^{(t)}$ corresponding to $s^{(t)}$ that is isometric to $S^{(0)}$ and $S^{(1)}$ if such a mesh exists.
\end{prop}

The algorithm can easily be extended to work for a non-connected triangular mesh $M$ by removing rigid transformations for each connected component of $M$ using local coordinate systems. We can then adapt the algorithm by finding the dual graph $D(M)$ and a minimum spanning tree $T(M)$ for each connected component of $M$. With this information, we can traverse the graph as described above. We denote this the \textit{exponential algorithm} in the following.

\section{Efficient Algorithm to Deform Triangular Meshes}
\label{efficient}

As the exponential algorithm is limited for pragmatic reasons to triangular meshes with few vertices, this section describes a more computationally efficient algorithm to find the deformed poses. 

To compute a unique triangular mesh $S^{(t)}$ given a point $s^{(t)} \in \mathcal{S}$ that linearly interpolates between $s^{(0)}$ and $s^{(1)}$, such that $S^{(t)}$ approximates the information given in $s^{(t)}$ well, we use the minimum spanning tree $T(M)$ of the dual graph $D(M)$ of $M$ as before.
An approach that reduces the total number of candidate coordinates of each vertex of $S^{(t)}$ to $O(n)$ is to restrict to one the number of times each edge of $T(M)$ can be traversed. We traverse $T(M)$ in depth-first order. When an edge $D(e)$ is traversed, we add candidate coordinates to one vertex $v$ as shown in Figure~\ref{num_candidates}. However, we only add at most a linear number of candidate coordinates to $v$ as described below. Denote the vertices of $e$ by $v_0(e)$ and $v_1(e)(e)$, denote the number of candidate coordinates that were added to $v_0(e)$ and $v_1(e)$, respectively, during the traversal of $T(M)$ before traversing the edge $D(e)$ by $d_1$ and $d_2$, and let $v_0(e)$ be the vertex of $e$ that was updated more recently in the traversal of $T(M)$. Let the candidate coordinates of $v_0(e)$ ($v_1(e)$, respectively) be given by $c_1^{1}, \ldots, c_1^{d_1}$ ($c_2^{1}, \ldots, c_2^{d_2}$, respectively) ordered from the least recently to the most recently added candidate coordinate. When traversing $D(e)$, we add $d_2$ candidate coordinates to $v$ by computing coordinates of $v$ based on the candidate pairs $(c_1^{d_1}, c_2^{1}), \ldots, (c_1^{d_1}, c_2^{d_2})$.

This strategy computes $O(n)$ candidates per vertex of $S^{(t)}$, and hence a total of $O(n^2)$ candidates, thereby avoiding the computation of an exponential number of candidate coordinates. To find the final coordinate of a vertex $v$, we average all of the candidate coordinates of $v$. We denote this the \textit{averaging algorithm} in the following.

The averaging approach was found to yield satisfactory results in all test cases. Just as in the algorithm by Kilian et al.~\cite{kilian_2007_gmss}, the results depend heavily on the initial rigid alignment of the shapes.

Furthermore, the results of the averaging algorithm depend on the order of the vertices and edges in $M$. Recall that the first vertex and its first incident edge are known. As we limit the number of times an edge of $T(M)$ can be traversed, not all of the information is propagated through the graph. Hence, the first vertex and the first edge influence the final result. We can eliminate this dependence at the cost of a higher running time by taking $O(n)$ shape spaces; one with each possible oriented edge as first edge. The running time of this algorithm is $O(n^3)$, which is too high for practical applications. We therefore can also use a smaller number of first edges that are found using Voronoi sampling~\cite{eldar_lindenbaum_porat_zeevi_94_farthest_point_image_sampling}. If $O(\log n)$ samples are used, the running time of the algorithm becomes $O(n \log n)$ and the influence of the first edge is expected to be low. We do not include experiments obtained using these extensions of the approach as the averaging algorithm yields similar results as the less efficient algorithms in all of our experiments.

\section{Experiments}
\label{experiments}

This section presents experiments using the algorithms presented in this paper. We show exact morphs of triangulated $3D$ polygons, morphs of models with few vertices computed using the exponential algorithm, and morphs of larger triangular meshes using the averaging algorithm.

The experiments were conducted using an implementation in C++ on an Intel (R) Pentium (R) D with 3.5 GB of RAM. OpenMP was used to improve the efficiency of the algorithms. To compute the minimum spanning tree $T(M)$, the boost graph library~\cite{boost} was used. 

\subsection{Deforming Triangulated $3D$ Polygons}

We demonstrate the efficient polygon algorithm using the simple polygon shown in Figure~\ref{step_deformation}. We deform the polygon shown in Figure~\ref{step_deformation} (a) to the polygon shown in Figure~\ref{step_deformation} (i). The morph is illustrated in Figures~\ref{step_deformation} (b)-(h). All of the intermediate poses are isometric to the start and end poses. The overlayed poses are shown in Figure~\ref{full_deformation}. The running time of this example is less than 1 second.

\begin{figure}[htb]
\centering
\includegraphics[width=5.0cm]{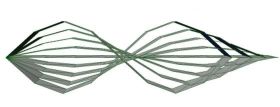}
\caption{\textit{Isometric morph of a simple polygon. The start polygon is a $3D$ polygon obtained by discretizing the curve $y=\sin(x)$ and by adding thickness to the curve along the $z$-direction. The end polygon is similarly obtained from $y = -\sin{x}$.}}
\label{full_deformation}
\end{figure}

\begin{figure}[htb]
\centering
\begin{tabular}{c c c c c c c c c}
\includegraphics[width = 1.3cm]{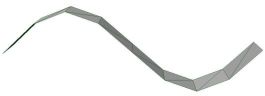} &
\includegraphics[width = 1.3cm]{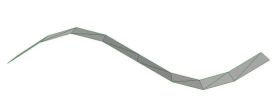} &
\includegraphics[width = 1.3cm]{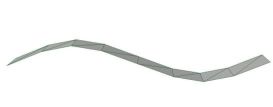} &
\includegraphics[width = 1.3cm]{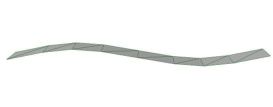} &
\includegraphics[width = 1.3cm]{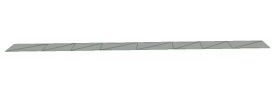} &
\includegraphics[width = 1.3cm]{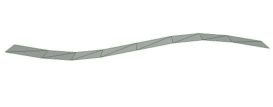} &
\includegraphics[width = 1.3cm]{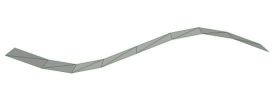} &
\includegraphics[width = 1.3cm]{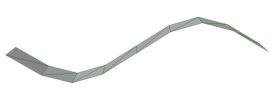} &
\includegraphics[width = 1.3cm]{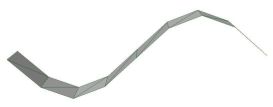} \\
(a) & (b) & (c) & (d) & (e) & (f) & (g) & (h) & (i) \\
\end{tabular}
\caption{\textit{Isometric morph of a simple polygon from pose (a) to pose (i). Intermediate poses obtained using the polygon algorithm are given from left to right.}}
\label{step_deformation}
\end{figure}

\subsection{Deforming General Triangular Meshes}

We present experimental results for the exponential algorithm and the averaging algorithm. First, we run the exponential algorithm on one model with few vertices to demonstrate the quality of the results. The model we use to test the approach is shown in Figure~\ref{step_deformation_ex_joe}. We aim to smoothly and isometrically deform the pose shown in Figure~\ref{step_deformation_ex_joe}(a) to the pose shown in Figure~\ref{step_deformation_ex_joe}(i). As mentioned above, there is no isometric deformation between the poses that interpolates the triangle normals. The result of our algorithm is shown in Figures~\ref{step_deformation_ex_joe}(b)-(h). Note that all triangle normals are interpolated and the symmetry of the model is preserved. Furthermore, all edge lengths with the exception of the edges of the four top faces are interpolated.

\begin{figure}[htb]
\centering
\begin{tabular}{c c c c c c c c c}
\includegraphics[width = 1.3cm]{ExampleJoe/pose1.jpg} &
\includegraphics[width = 1.3cm]{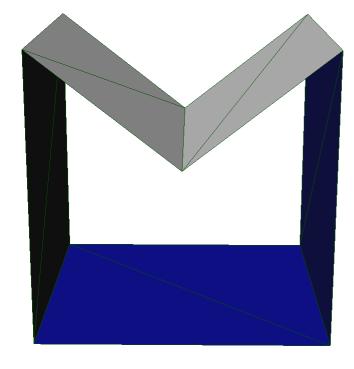} &
\includegraphics[width = 1.3cm]{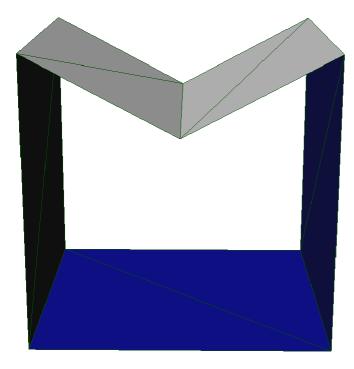} &
\includegraphics[width = 1.3cm]{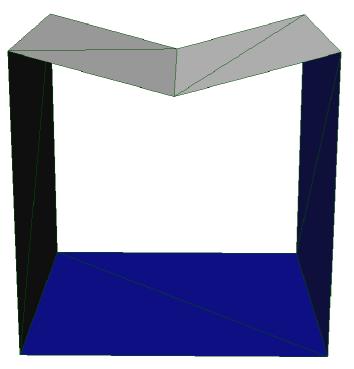} &
\includegraphics[width = 1.3cm]{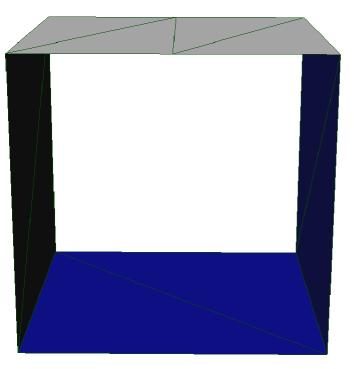} &
\includegraphics[width = 1.3cm]{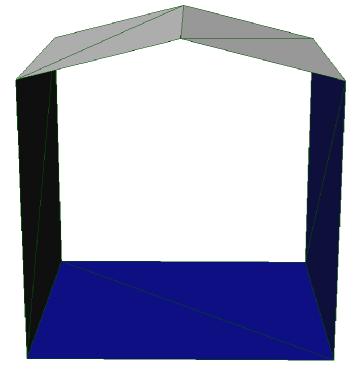} &
\includegraphics[width = 1.3cm]{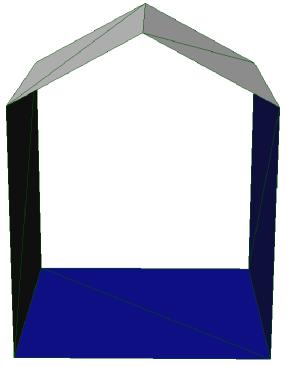} &
\includegraphics[width = 1.3cm]{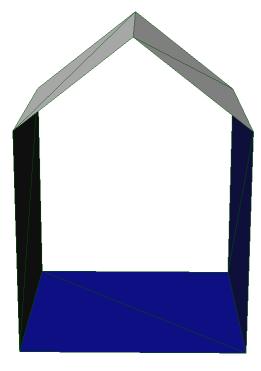} &
\includegraphics[width = 1.3cm]{ExampleJoe/pose9.jpg} \\
(a) & (b) & (c) & (d) & (e) & (f) & (g) & (h) & (i) \\
\end{tabular}
\caption{\textit{Isometric morph of a cycle from pose (a) to pose (i). Intermediate poses obtained using the exponential algorithm are given from left to right.}}
\label{step_deformation_ex_joe}
\end{figure}

Second, we demonstrate the quality and efficiency of the averaging algorithm. The first experiment morphs between animated poses of the armadillo model. The models are chosen from the AIM@SHAPE repository~\footnote{http://shapes.aimatshape.net/releases.php}. The models contain $331904$ triangles and $165954$ vertices. Testing our algorithms on these models emphasizes the practical use of our method. The results are shown in Figure~\ref{armadillo}. Note that the intermediate poses are visually pleasing. Although some small bumps appear on the right arm of the armadillo in Figure~\ref{armadillo} (b) and on the left arm of the armadillo in Figure~\ref{armadillo} (e), most details are preserved during the morph.

\begin{figure}[htb]
\centering
\begin{tabular}{c c c}
\includegraphics[height = 4.5cm]{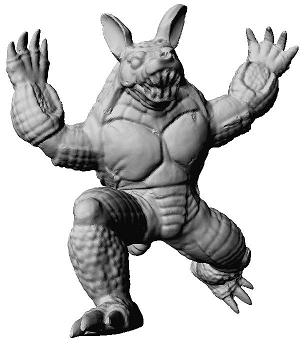} &
\includegraphics[height = 4.5cm]{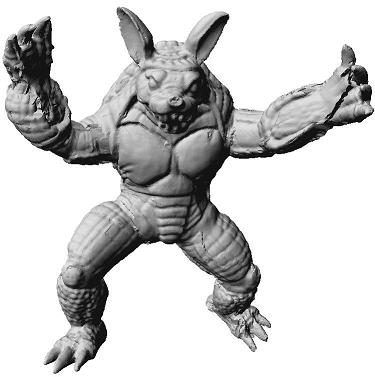} &
\includegraphics[height = 4.5cm]{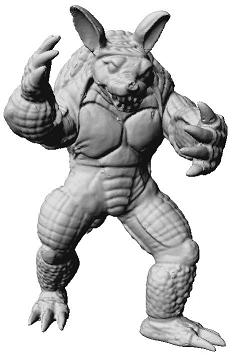} \\
(a) & (b) & (c)\\
\includegraphics[height = 4.5cm]{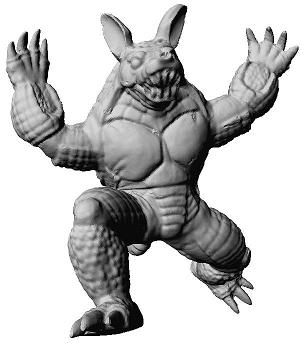} &
\includegraphics[height = 5.5cm]{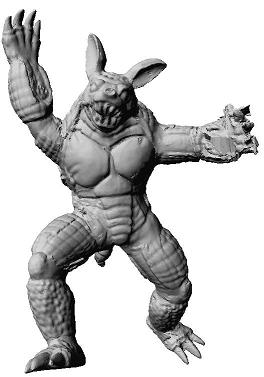} &
\includegraphics[height = 6.0cm]{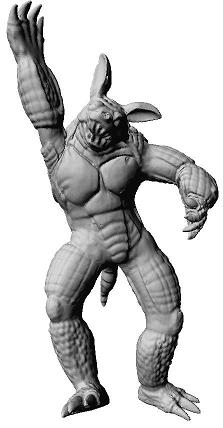} \\
(d) & (e) & (f)\\
\end{tabular}
\caption{\textit{Isometric morph of the armadillo model from pose (a) to pose (c) and from pose (d) to pose (f). The intermediate poses obtained using the averaging algorithm for $t=0.5$ are given in figures (b) and (e).}}
\label{armadillo}
\end{figure}

Figure~\ref{human} shows the morph between two poses of a model of a human being. The deformed pose of the model was found using the automatic technique by Baran and Popovi\'{c}~\cite{baran_popovic_07_animation}. The models contain $10002$ vertices. The two given meshes are shown in Figures~\ref{human}(a) and (i). Intermediate poses using the averaging algorithm are shown in Figures~\ref{human}(b) to (h). Note that although small bumps appear on the leg of the morphed surface, the morphs are intuitive. 

\begin{figure}[htb]
\centering
\begin{tabular}{c c c c c c c c c}
\includegraphics[height = 4.5cm]{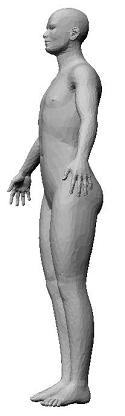} &
\includegraphics[height = 4.5cm]{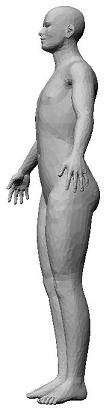} &
\includegraphics[height = 4.5cm]{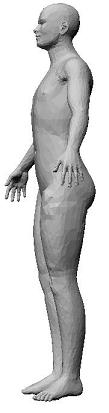} &
\includegraphics[height = 4.5cm]{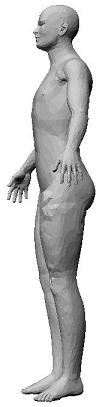} &
\includegraphics[height = 4.5cm]{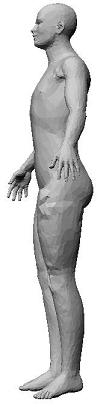} &
\includegraphics[height = 4.5cm]{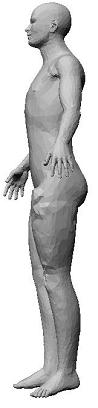} &
\includegraphics[height = 4.5cm]{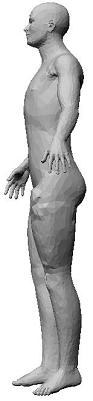} &
\includegraphics[height = 4.5cm]{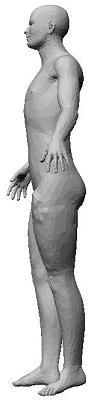} &
\includegraphics[height = 4.5cm]{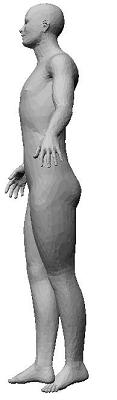} \\
(a) & (b) & (c) & (d) & (e) & (f) & (g) & (h) & (i) \\
\end{tabular}
\caption{\textit{Isometric morph of a human from pose (a) to pose (i). Intermediate poses obtained using the averaging algorithm are given from left to right.}}
\label{human}
\end{figure}

Furthermore, we morph between heads from the CAESAR database~\cite{robinette_daanen_paquet_99_caesar}. The correspondence between the two head models was found using the approach by Xi et al.~\cite{xi_lee_shu_07_bodies}. The models contain 11102 vertices. The two given meshes are shown in Figures~\ref{head}(a) and (i). Intermediate poses obtained using the averaging algorithm are shown in Figure~\ref{head} (b) to (h). Note that the morphs are visually pleasing. 

\begin{figure}[htb]
\centering
\begin{tabular}{c c c c c c c c c}
\includegraphics[width = 1.3cm]{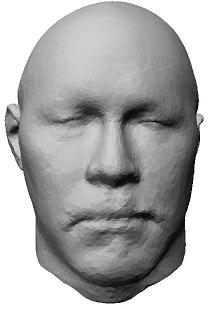} &
\includegraphics[width = 1.3cm]{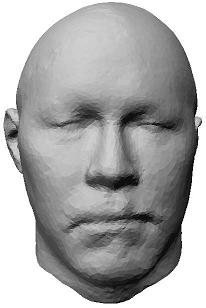} &
\includegraphics[width = 1.3cm]{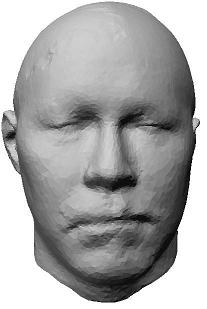} &
\includegraphics[width = 1.3cm]{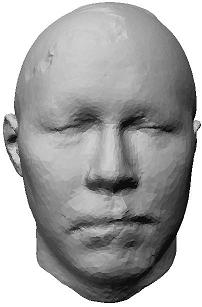} &
\includegraphics[width = 1.3cm]{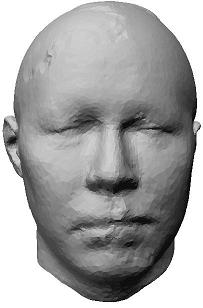} &
\includegraphics[width = 1.3cm]{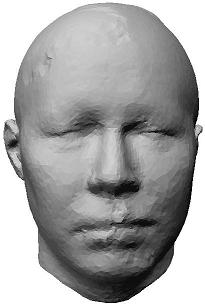} &
\includegraphics[width = 1.3cm]{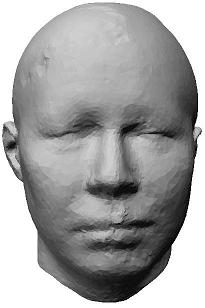} &
\includegraphics[width = 1.3cm]{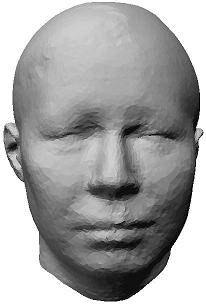} &
\includegraphics[width = 1.3cm]{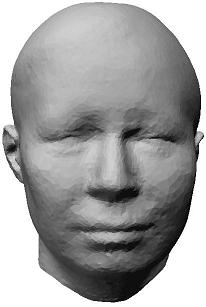} \\
(a) & (b) & (c) & (d) & (e) & (f) & (g) & (h) & (i) \\
\end{tabular}
\caption{\textit{Isometric morph of a head from pose (a) to pose (i). Intermediate poses obtained using the averaging algorithm are given from left to right.}}
\label{head}
\end{figure}

Finally, we morph between a few poses of the Alien model shown in Figure~\ref{aliens_inv_avg}. The Alien model is chosen from the Princeton Shape Benchmark\footnote{http://shape.cs.princeton.edu/benchmark/} and animated to obtain multiple postures with known correspondences using the automatic technique by Baran and Popovi\'{c}~\cite{baran_popovic_07_animation}. All of the Alien models contain $429$ vertices. The results are shown in Figure~\ref{aliens_inv_avg} and we can see that all of the morphs are visually pleasing. 

\begin{figure}[htb]
\centering
\includegraphics[width = 8.0cm]{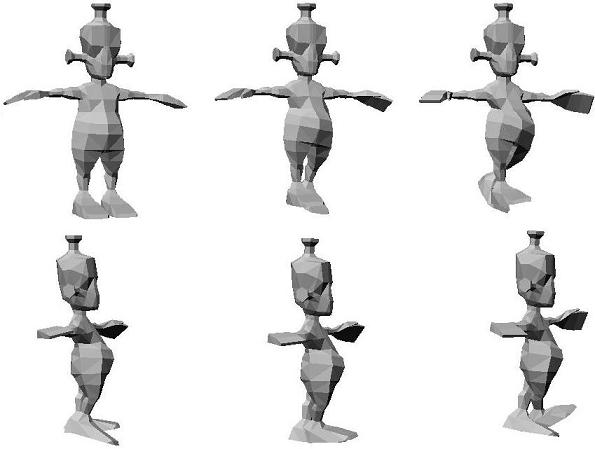} 
\caption{\textit{The two rows show start and end poses on the left and right, respectively. The intermediate poses shown in the middle columns are interpolations for $t=0.5$ obtained using the averaging algorithm. }}
\label{aliens_inv_avg}
\end{figure}

For all of the experiments conducted, we measured the time efficiency to compute one intermediate pose at $t=0.5$ and the energy 
$$Q = \sum_{e\in E} (\left\|v_0(e)-v_1(e)\right\| - l(e))^2,$$ 
where $E$ is the set of edges of $M$, $v_0(e)$ and $v_1(e)$ are the vertices of $e$ in $S^{(t)}$, $\left\|a, b\right\|$ denotes the Euclidean distance between $a$ and $b$, and $l(e)$ is the desired length for $e$ stored in $s^{(t)}$. The running times and the quality of the results are summarized in Table~\ref{morphing_table}. Note that the developed implementation is non-optimized and experimental. The running times might be improved by implementing some of the algorithms on the GPU.

\begin{table}
\centering
\begin{tabular}{|l|r|r|r|r|r|r|}
\hline
&Alien Row 1 & Alien Row 2& Human & Head & Armadillo Row 1 & Armadillo Row 2\\
\hline
$n$& 429 & 429& 10002 & 11102 &165954 &165954\\
\hline
time (sec)& $< 1$& $< 1$& 8& 10& 1813&1763\\
\hline
$Q$& 0.000719& 0.000064& 0.042674& 0.003583&12277.92 &70637.80\\
\hline
\end{tabular}
\caption{\textit{Summary of the experimental results.}}
\label{morphing_table}
\end{table}

\section{Conclusion}

We presented an approach to morph efficiently between isometric poses of triangular meshes in a novel shape space. The main advantage of this morphing method is that the most isometric morph is always found in linear time when triangulated $3D$ polygons are considered. For general triangular meshes, the approach cannot be proven to find the optimal solution. However, this paper presents an efficient heuristic approach to find a morph for general triangular meshes that does not depend on solving a non-linear optimization problem.

The presented experimental results demonstrate that the heuristic approach yields visually pleasing results. The approach is not invariant with respect to the order of the vertices of the mesh, but can be modified to have this property at the cost of a higher running time. 

An interesting direction for future work is to find an efficient way of morphing triangular meshes while guaranteeing that no self-intersections occur. For polygons in two dimensions, this problem was solved using an approach based on energy minimization~\cite{iben_obrien_demaine_06_refolding_polygons}.

\section*{Acknowledgments}
We thank Martin Kilian for sharing his insight on the topic of shape spaces with us. We thank Pengcheng Xi for providing us the data for the head experiment.


\end{document}